\pdfoutput=1
\RequirePackage{ifpdf}
\ifpdf 
\documentclass[pdftex]{sigma}
\else
\documentclass{sigma}
\fi

\numberwithin{equation}{section}

\newtheorem{Theorem}{Theorem}[section]
\newtheorem{Lemma}[Theorem]{Lemma}
\newtheorem{Proposition}[Theorem]{Proposition}
{ \theoremstyle{definition}
\newtheorem{Remark}[Theorem]{Remark} }

\begin{document}

\allowdisplaybreaks

\newcommand{\arXivNumber}{1801.07312}

\renewcommand{\PaperNumber}{023}

\FirstPageHeading

\ShortArticleName{On the Linearization Covering Technique}

\ArticleName{On the Linearization Covering Technique\\ and its Application to Integrable Nonlinear\\ Differential Systems}

\Author{Anatolij K.~PRYKARPATSKI~$^{\dag\ddag}$}

\AuthorNameForHeading{A.K.~Prykarpatski}

\Address{$^\dag$~The Department of Physics, Mathematics and Computer Science,\\
\hphantom{$^\dag$}~Cracow University of Technology, Krak\'ow 30-155, Poland}
\Address{$^\ddag$~Ivan Franko State Pedagogical University of Drohobych, Lviv Region, Ukraine}
\EmailD{\href{mailto:pryk.anat@cybergal.com}{pryk.anat@cybergal.com}}

\ArticleDates{Received January 22, 2018, in final form February 28, 2018; Published online March 16, 2018}

\Abstract{In this letter I analyze a covering jet manifold scheme, its relation to the invariant theory of the associated vector fields, and applications to the Lax--Sato-type integrability of nonlinear dispersionless differential systems. The related contact geometry linearization covering scheme is also discussed. The devised techniques are demonstrated for such nonlinear Lax--Sato integrable equations as Gibbons--Tsarev, ABC, Manakov--Santini and the differential Toda singular manifold equations.}


\Classification{17B68; 17B80; 35Q53; 35G25; 35N10; 37K35; 58J70; 58J72; 34A34; 37K05; 37K10}

\section{Introductory notions and examples}\label{Sec_Th1}

Some three years ago I.~Krasil'shchik in his work \cite{Kras} has analyzed the so-called Gibbons--Tsarev equation
\begin{gather}
z_{yy}+z_{t}z_{ty}-z_{y}z_{tt}+1=0, \label{T0}
\end{gather}
and its so-called nonlinear first-order differential \textit{covering}
\begin{gather}
\frac{\partial w}{\partial t}-\frac{1}{z_{y}+z_{t}w-w^{2}}=0,\qquad \frac{\partial w}{\partial y}+\frac{z_{t}-w}{u_{y}+z_{t}w-w^{2}}=0, \label{T0a}
\end{gather}
which for any solution $z\colon \mathbb{R}^{2}\rightarrow\mathbb{R}$ to equation~(\ref{T0}) is compatible for all $(t,y)\in \mathbb{R}^{2}$. He has shown that the latter makes it possible to determine for any smooth solution $w\colon \mathbb{R}^{2}\rightarrow\mathbb{R}$ to the equation~(\ref{T0a}) suitable smooth functions $\psi\colon \mathbb{R}\times\mathbb{R}^{2}\rightarrow \mathbb{R}$ satisfying the associated linear representation of Lax--Sato type:
\begin{gather*}
\frac{\partial\psi}{\partial t}+\frac{1}{z_{y}+z_{t}\lambda-\lambda^{2}} \frac{\partial\psi}{\partial\lambda} =0, \qquad \frac{\partial\psi}{\partial y}-\frac{z_{t}-\lambda}{z_{y}+z_{t}\lambda-\lambda^{2}}\frac{\partial\psi }{\partial\lambda}=0, 
\end{gather*}
which for any solution to the equation (\ref{T0}) is also compatible for all $(t,y)\in \mathbb{R}^{2}$ and an arbitrary parameter $\lambda\in \mathbb{R}$.

In his work \cite{Kras} I.~Krasil'shchik also posed an interesting problem of providing a differential-geometric explanation of the linearization procedure for a given nonlinear differential-geometric relation $J^{1}(\mathbb{R}^{n}\times\mathbb{R}^{2};\mathbb{R})|_{\mathcal{E}}$ in the jet manifold $J^{1}\big(\mathbb{R}^{n}\times\mathbb{R}^{2};\mathbb{R}\big)$, $n\in\mathbb{Z}_{+}$, realizing a compatible \textit{covering}~\cite{BaKrMoVo} for the associated nonlinear differential equation $\mathcal{E}[x,\tau;u]=0$, imbedded into some jet manifold $J^{k}\big(\mathbb{R}^{n}\times\mathbb{R}^{2};\mathbb{R}^{m}\big)$ for some $k,m\in\mathbb{Z}_{+}$. His extended version of explanation of this
procedure, presented in~\cite{Kras}, appeared abstract enough and held it in a hidden form with no new example demonstrating its application. From this point of view, our work seeks to unveil some important points of this linearization procedure in the framework of the classical inhomogeneous vector field equations and present its new and important applications. We consider the jet manifold $J^{k}\big(\mathbb{R}^{n}\times\mathbb{R}^{2}; \mathbb{R}^{m}\big)$ for some fixed $k,m\in\mathbb{Z}_{+}$ and some differential relation~\cite{Grom} in a~general form $\mathcal{E}[x,\tau;u]=0$, satisfied for all $(x;\tau)\in$ $\mathbb{R}^{n}\times\mathbb{R}^{2}$ and suitable smooth mappings $u\colon \mathbb{R}^{n}\times\mathbb{R}^{2}\rightarrow \mathbb{R}^{m}$.

As a new and interesting example, we can take $n=1=m$, $k=2$ and choose a~differential relation $\mathcal{E}[x;y,t;u]=0$ of the form
\begin{gather}
u_{t}u_{xy}-k_{1}u_{x}u_{ty}-k_{2}u_{y}u_{tx}=0, \label{T1}
\end{gather}
the so-called ABC equation, or the Zakharevich equation, first discussed in \cite{Zakh} for $k_{1}+k_{2}-1=0$, where $u\colon \mathbb{R}\times\mathbb{R}^{2}\rightarrow\mathbb{R}$ and $k_{1},k_{2}\in\mathbb{R}$ are arbitrary parameters satisfying the conditions
\begin{gather*}
k_{1}+k_{2}-1=0 \vee k_{1}+k_{2}-1\neq0. 
\end{gather*}
The first case $k_{1}+k_{2}-1=0$ was before intensively investigated in \cite{DuKr,MoSe,Zakh} and recently in \cite{HePrBlPr}, where there was found its Lax--Sato type linearization and many other its interesting mathematical properties were revealed.
For the second case $k_{1}+k_{2}-1\neq0$ the following crucial result was stated in equivalent form by P.A.~Burovskiy, E.V.~Ferapontov and S.P.~Tsarev in \cite{BuFeTs} and recently by I.~Krasil'shchik, A.~Sergyeyev and O.~Morozov in~\cite{KrSeMo}.

\begin{Proposition}
The covering system $J^{1}\big(\mathbb{R}\times\mathbb{R}^{2};\mathbb{R}\big)|_{\mathcal{E}}$ of quasi-linear first-order differential
relations
\begin{gather}
\frac{\partial w}{\partial t}+\frac{u_{t}w}{u_{x}k_{1}(k_{1}+k_{2}-1)}\frac{\partial w}{\partial x}-\frac{w(w+k_{1}+k_{2}-1)u_{tx}}{u_{x}k_{1}}
=0,\nonumber\\
\frac{\partial w}{\partial y}+\frac{u_{y}w}{u_{x}k_{1}(w+k_{1}+k_{2}-1)}\frac{\partial w}{\partial x}-\frac{w(k_{1}+k_{2}-1)u_{yx}}{u_{x}k_{1}}
 =0\label{T3}
\end{gather}
on the jet manifold $J^{1}\big(\mathbb{R}\times\mathbb{R}^{2};\mathbb{R}\big)$ is compatible, that is, it holds for any its smooth solution $w\colon \mathbb{R} \times\mathbb{R}^{2}\rightarrow\mathbb{R}$ at all points $(x;y,t)\in \mathbb{R}\times\mathbb{R}^{2}$, iff the function $u\colon \mathbb{R}\times\mathbb{R}^{2}\rightarrow\mathbb{R}$ satisfies the ABC-equation~\eqref{T1}.
\end{Proposition}

Moreover, this result was recently generalized in \cite{KrSeMo} to the following ``\textit{linearizing}'' proposition.

\begin{Proposition}\label{Prop_T1}A system $J_{\rm lin}^{1}\big(\mathbb{R}^{2}\times\mathbb{R}^{2};\mathbb{R}\big)|_{\mathcal{E}}$ of linear first-order differential relations
\begin{gather}
\frac{\partial\psi}{\partial t}+\frac{\lambda u_{x}^{\frac{k_{2}-1}{k_{1}}
}u_{t}}{k_{1}(k_{1}+k_{2}-1)}\frac{\partial\psi}{\partial x}-\frac
{\lambda^{2}u_{x}^{\frac{k_{2}-k_{1}-1}{k_{1}}}(u_{t}u_{xx}-k_{1}
u_{xx}u_{xt})}{k_{1}^{2}}\frac{\partial\psi}{\partial\lambda}=0,\nonumber\\
\frac{\partial\psi}{\partial y}+\frac{k_{2}\lambda u_{x}^{\frac{k_{2}-1}
{k_{1}}}u_{y}}{k_{1}\big(\lambda u_{x}^{\frac{k_{2}-k_{1}-1}{k_{1}}}+k_{1}
+k_{2}-1\big)}\frac{\partial\psi}{\partial x}-\frac{\lambda^{2}k_{2}(k_{1}
+k_{2}-1)u_{x}^{\frac{k_{2}-k_{1}-1}{k_{1}}}u_{y}u_{xx}}{k_{1}^{2}\big(\lambda
u_{x}^{\frac{k_{2}-k_{1}-1}{k_{1}}}+k_{1}+k_{2}-1\big)}\frac{\partial\psi
}{\partial\lambda}=0\label{T4}
\end{gather}
on the covering jet manifold $J^{1}\big(\mathbb{R}^{2}\times\mathbb{R}^{2};\mathbb{R}\big)$ is compatible, that is, it holds for any its smooth
solution $\psi\colon \mathbb{R}^{2}\times\mathbb{R}^{2}\rightarrow\mathbb{R}$ at all points $(x,\lambda;y,t)\in\mathbb{R}^{2}\times\mathbb{R}^{2}$, iff the function $u\colon \mathbb{R}\times\mathbb{R}^{2}\rightarrow\mathbb{R}$ satisfies the ABC-equation~\eqref{T1}.
\end{Proposition}

A similar result, when $n=1$, $m,k=2$, we have stated for the Manakov--Santini equations%
\begin{gather}
u_{tx}+u_{yy}+(uu_{x})_{x}+v_{x}u_{xy}-v_{y}u_{xx} =0, \nonumber\\
v_{xt}+v_{yy}+uv_{xx}+v_{x}u_{xy}-v_{y}v_{xx} =0,\label{T4aa}
\end{gather}
whose Lax--Sato integrability was extensively studied earlier in \cite{BoPa, DuFeKr,FeKr,HePrBlPr,MaSa,OdSo}. The system~(\ref{T4aa}) as a
jet submanifold $J^{1}\big(\mathbb{R}\times\mathbb{R}^{2};\mathbb{R}\big)|_{\mathcal{E}}\subset J^{1}(\mathbb{R}\times\mathbb{R}^{2};\mathbb{R})$ allows the following nonlinear first-order differential Lax-type representation
\begin{gather}
\frac{\partial w}{\partial t}+\big(w^{2}-wv_{x}+u-v_{y}\big)\frac{\partial w}{\partial x}+u_{x}w-u_{y}+v_{yy}+v_{x}(v_{y}-u)_{x}=0,\nonumber \\
\frac{\partial w}{\partial y}+w\frac{\partial w}{\partial x}-v_{xx}w+(u-v_{y})_{x}=0,\label{T4b}
\end{gather}
compatible on solutions to the nonlinear differential relation $\mathcal{E}[x,\tau;u]=0$ (\ref{T4aa}) on $J^{2}\big(\mathbb{R}\times \mathbb{R}^{2};\mathbb{R}^{2}\big)$. The existence of the compatible representation~(\ref{T4b}) makes it possible to state the following proposition.

\begin{Proposition}\label{Prop_T2}A covering system $J_{\rm lin}^{1}\big(\mathbb{R}^{2}\times \mathbb{R}^{2};\mathbb{R}\big)|_{\mathcal{E}}$ of linear first-order differential relations
\begin{gather*}
\frac{\partial\psi}{\partial t}+\big(\lambda^{2}+\lambda v_{x}+u-v_{y}\big)\frac{\partial\psi}{\partial x}+(u_{y}-\lambda u_{x})\frac{\partial\psi
}{\partial\lambda}=0,\nonumber\\
\frac{\partial\psi}{\partial y}+(v_{x}+\lambda)\frac{\partial\psi}{\partial x}-u_{x}\frac{\partial\psi}{\partial\lambda}=0
\end{gather*}
on the jet manifold $J^{1}\big(\mathbb{R}^{2}\times\mathbb{R}^{2};\mathbb{R}\big)$ is compatible, that is, it holds for any its smooth solution $\psi\colon \mathbb{R}^{2}\times\mathbb{R}^{2}\rightarrow\mathbb{R}$ at all points $(x,\lambda;y,t)\in \mathbb{R}^{2}\times\mathbb{R}^{2}$, iff the function $u\colon \mathbb{R}\times\mathbb{R}^{2}\rightarrow\mathbb{R}$ satisfies the Manakov--Santini equation~\eqref{T4aa}.
\end{Proposition}

From the point of view, based on the propositions above, the substance of our present analysis, similarly to that suggested in \cite{Kras}, is to
recover the intrinsic mathematical structure responsible for the existence of the ``\textit{linearizing}'' covering jet manifold mappings
\begin{gather}
J_{\rm lin}^{1}\big(\mathbb{R}^{n+1}\times\mathbb{R}^{2};\mathbb{R}\big)|_{\mathcal{E}}
 \rightleftharpoons J^{1}\big(\mathbb{R}^{n}\times\mathbb{R}^{2};\mathbb{R}\big)|_{\mathcal{E}} \label{T4a}
\end{gather}
for any dimension $n\in\mathbb{Z}_{+}$, compatible with our differential relation $\mathcal{E}[x,\tau;u]=0$, as it was presented above in the
form~(\ref{T3}) and~(\ref{T4}) for the differential relation (\ref{T1}). Thus, for a given nonlinear differential relation $\mathcal{E}[x,\tau;u]=0$ on the jet manifold $J^{k}\big(\mathbb{R}^{n}\times\mathbb{R}^{2};\mathbb{R}^{m}\big)$ for some $k\in\mathbb{Z}_{+}$ one can pose the
following \textit{problem}:

Suppose there holds a compatible system $J^{1}\big(\mathbb{R}^{n}\times\mathbb{R}^{2};\mathbb{R}^{m}\big)|_{\mathcal{E}}\subset J^{1}\big(\mathbb{R}^{n}\times \mathbb{R}^{2};\mathbb{R}^{m}\big)$ of quasi-linear first-order differential relations. How then can one construct a linearizing first-order differential system $J_{\rm lin}^{1}\big(\mathbb{R}^{(1+n)+2};\mathbb{R}\big)|_{\mathcal{E}}\subset J^{1}\big(\mathbb{R}^{(1+n)+2};\mathbb{R}\big)$ in the form of equations involving vector fields on the covering space $\mathbb{R}^{n+1}\times\mathbb{R}^{2} $, implementing the implications~(\ref{T4a})?

The latter is interpreted as the associated Lax--Sato representation \cite{ZaSh}, recently developed in \cite{Bogd,BoDrMa,BoKo,FeKr,FeKrNo,TaTa-1,TaTa-2}, for the given differential relation $\mathcal{E}[x,\tau;u]=0$ on the jet manifold $J^{k}\big(\mathbb{R}^{n}\times\mathbb{R}^{2};\mathbb{R}^{m}\big)$.

As a dual approach to this linearization covering scheme we present also the so-called contact geometry linearization, proposed recently in~\cite{Serg-1} and generalizing from three to four dimensions the well-known, see, e.g.,~\cite{SzBl} and reference therein, Hamiltonian linearization covering method. As an example, we state the following proposition.

\begin{Proposition}
The following {\rm \cite{BuFeTs}} nonlinear singular manifold Toda differential relation
\begin{gather}
u_{xy}\sinh^{2}u_{t}=u_{x}u_{y}u_{tt} \label{T4ab}
\end{gather}
on the jet manifold $J^{2}\big(\mathbb{R}^{2}\times\mathbb{R}^{2};\mathbb{R}\big)$ allows the Lax--Sato type linearization covering
\begin{gather*}
\frac{\partial\psi}{\partial t}+\frac{\big(e^{-2u_{t}}-1\big)}{2u_{x}}\frac
{\partial\psi}{\partial x}-\left[ \lambda\left( \frac{e^{-2u_{t}}-1}{2u_{x}
}\right) _{x}+\lambda^{2}\left( \frac{e^{-2u_{t}}-1}{2u_{x}}\right)
_{z}\right] \frac{\partial\psi}{\partial\lambda} =0,\nonumber\\
\frac{\partial\psi}{\partial y}-\frac{u_{y}e^{-2u_{t}}}{u_{x}}\frac
{\partial\psi}{\partial x}+\left[ \lambda\left( \frac{u_{y}e^{-2u_{t}}
}{u_{x}}\right) _{x}+\lambda^{2}\left( \frac{u_{y}e^{-2u_{t}}}{u_{x}
}\right) _{z}\right] \frac{\partial\psi}{\partial\lambda} =0
\end{gather*}
for smooth invariant functions $\psi\in C^{2}\big(\mathbb{R}^{3}\times \mathbb{R}^{2};\mathbb{R}\big)$, all $(x,z,\lambda;\tau)\in\mathbb{R}^{3}
\times\mathbb{R}^{2}$ and any smooth solution $u\colon \mathbb{R}^{2}\times\mathbb{R}^{2}\mathbb{\rightarrow R}$ to the differential relation~\eqref{T4ab}.
\end{Proposition}

\section{The linearization covering scheme}\label{Sec_T2}

Implementation of the scheme (\ref{T4a}) is based on the notion of invariants of suitably specified vector fields on the extended base space $\mathbb{R}^{n+1}\times\mathbb{R}^{2}$, whose definition suitable for our needs is as follows: a smooth mapping $\psi\colon \mathbb{R}^{n+1}\times \mathbb{R}^{2}\rightarrow\mathbb{R}$ is an invariant of a set of vector fields
\begin{gather*}
X^{(k)}:=\frac{\partial}{\partial\tau_{k}}+\sum_{j=1}^n a_{j}^{(k)}(x,\lambda;\tau)\frac{\partial}{\partial x_{j}}
+b^{(k)}(x,\lambda ;\tau)\frac{\partial}{\partial\lambda} 
\end{gather*}
on $\mathbb{R}^{n+1}\times\mathbb{R}^{2}$ involving parameters $\tau \in\mathbb{R}^{2} $ with smooth coefficients $\big(a^{(k)},b^{(k)}\big)\colon \mathbb{R}^{n+1}\times\mathbb{R}^{2}\rightarrow \mathbb{E}^{n}\times\mathbb{R}$, $k=1,2$, if the conditions
\begin{gather}
X^{(k)}\psi=0 \label{T6}
\end{gather}
hold for $k=1,2$ and all $(x,\lambda;\tau)\in\mathbb{R}^{n+1}\times\mathbb{R}^{2}$. The system of linear equations (\ref{T6}) is equivalently representable as a jet submanifold $J_{\rm lin}^{1}\big(\mathbb{R}^{n+1}\times\mathbb{R}^{2};\mathbb{R}\big)|_{\mathcal{E}}\subset J^{1}\big(\mathbb{R}^{n+1}\times\mathbb{R}^{2};\mathbb{R}\big)$. It is also well known~\cite{Cart} that simultaneously the following vector field flows
\begin{gather}
\frac{\partial x_{j}}{\partial\tau_{k}} =a_{j}^{(k)}(x,\lambda;\tau),\qquad \frac{\partial\lambda}{\partial\tau_{k}} =b^{(k)}(x,\lambda;\tau)
\label{T7}
\end{gather}
are compatible for any $j=1,\dots,n$, $k=1,2$ and all $(x,\lambda;\tau)\in\mathbb{R}^{n+1}\times\mathbb{R}^{2}$ too. Taking now into account that there is such an invariant function $\psi\colon \mathbb{R}^{n+1}\times\mathbb{R}^{2}\rightarrow\mathbb{R}$, which can be represented as $\psi(x,\lambda;\tau)=w(x;\tau)-\lambda:=0$ for some smooth mapping $w\colon \mathbb{R}^{n}\times\mathbb{R}^{2}\rightarrow\mathbb{R}$, it gives upon its substitution into~(\ref{T6}) the following \textit{a priori} compatible reduced system of quasi-linear first-order
differential relations
\begin{gather}
\frac{\partial w}{\partial\tau_{k}}+\sum_{j=1}^n
a_{j}^{(k)}(x,w;\tau)\frac{\partial w}{\partial x_{j}}-b^{(k)}(x,w;\tau)=0
\label{T8}
\end{gather}
for $k=1,2$ on the jet manifold $J^{0}\big(\mathbb{R}^{n}\times \mathbb{R}^{2};\mathbb{R}\big)$. Moreover, subject to the system~(\ref{T8})
one can also observe~\cite{Cart,Godb} that, modulo solutions to the vector field flow evolution equations~(\ref{T7}), the expression $w(x;\tau)= \psi(x,\lambda(\tau);\tau)+\lambda(\tau)$ for all $(x;\tau)\in\mathbb{R}^{n}\times\mathbb{R}^{2}$, where $\psi\colon \mathbb{R}^{n+1} \times\mathbb{R}^{2}\rightarrow\mathbb{R}$ is a first integral of the vector field flows (\ref{T7}). Thus, the reduction scheme just described above provides the algorithm
\begin{gather*}
J_{\rm lin}^{1}\big(\mathbb{R}^{n+1}\times\mathbb{R}^{2};\mathbb{R}\big)|_{\mathcal{E}}
\rightarrow J^{1}\big(\mathbb{R}^{n}\times\mathbb{R}^{2};\mathbb{R}\big)|_{\mathcal{E}}
\end{gather*}
from the implications (\ref{T4a}) formulated above. The corresponding \textit{inverse} implication{\samepage
\begin{gather}
J_{\rm lin}^{1}\big(\mathbb{R}^{n+1}\times\mathbb{R}^{2};\mathbb{R}\big)|_{\mathcal{E}}
\leftarrow J^{1}\big(\mathbb{R}^{n}\times\mathbb{R}^{2};\mathbb{R}\big)|_{\mathcal{E}}\label{T10}
\end{gather}
can be algorithmically described as follows.}

Consider now a compatible system $J^{1}\big(\mathbb{R}^{n+2};\mathbb{R}\big)|_{\mathcal{E}}\subset J^{1}\big(\mathbb{R}^{n+2};\mathbb{R}\big)$ of
the first-order nonlinear differential relations
\begin{gather}
\frac{\partial w}{\partial\tau_{k}}+\sum_{j=1}^n a_{j}^{(k)}(x,w;\tau)\frac{\partial w}{\partial x_{j}}-b^{(k)}(x,w;\tau)=0\label{T11}
\end{gather}
with smooth coefficients $\big(a^{(k)},b^{(k)}\big)\colon \mathbb{R}^{n+1}\times \mathbb{R}^{2}\rightarrow\mathbb{E}^{n}\times\mathbb{R}$, $k= 1,2 $. As the first step it is necessary to check \textit{whether the associated system of vector field flows}
\begin{gather}
\frac{\partial x_{j}}{\partial\tau_{k}} =a_{j}^{(k)}(x,w;\tau)\label{T12}
\end{gather}
on $\mathbb{R}^{n+1}$ modulo the flows (\ref{T11}) for all $j=1,\dots,n$ and $k=1,2$ \textit{is compatible too.} If the answer is ``\textit{yes}'', then it just means \cite{Cart} that any solution to (\ref{T11}) as a complex function $w\colon \mathbb{R}^{n}\times\mathbb{R}^{2} \rightarrow\mathbb{R}$ is representable as $w(x;\tau)-\lambda=\alpha (\psi(x,\lambda;\tau))$ for any $\lambda\in\mathbb{R}$ and some smooth mapping $\alpha\colon \mathbb{R}\rightarrow\mathbb{R}$, where the mapping $\psi \colon \mathbb{R}^{n+1}\times\mathbb{R}^{2}\rightarrow\mathbb{R} $ is \textit{a~first integral} of the vector field equations
\begin{gather}
\frac{\partial\psi}{\partial\tau_{k}}+\sum_{j=1}^n a_{j}^{(k)}(x,\lambda;\tau)\frac{\partial\psi}{\partial x_{j}}
+b^{(k)}(x,\lambda;\tau)\frac{\partial\psi}{\partial\lambda}=0 \label{T13}
\end{gather}
on the extended space $\mathbb{R}^{n+1}\times\mathbb{R}$ for all $(x,\lambda;\tau)\in \mathbb{R}^{n+1}\times\mathbb{R}^{2}$. Moreover, the
value $w(x(\tau);\tau)=\lambda\in\mathbb{R}$ for all $\tau\in\mathbb{R}^{2}$ is constant as it follows from the condition $\alpha(\psi(x(\tau),w;
\tau))=0$ for the $\tau\in \mathbb{R}^{2}$. Thus, we derived that the vector field equations~(\ref{T13}) realize the covering linear first-order
differential relations $J_{\rm lin}^{1}\big(\mathbb{R}^{n+1}\times\mathbb{R}^{2}; \mathbb{R}\big)|_{\mathcal{E}}\subset J^{1}\big(\mathbb{R}^{n+1}\times\mathbb{R}^{2};\mathbb{R}\big)$ for $k=1,2$ and all $(x,\lambda;\tau )\in
\mathbb{R}^{n+1}\times\mathbb{R}^{2}$, linearizing the first-order nonlinear differential relations~(\ref{T11}) and interpreting it as the
corresponding Lax--Sato representation.

On the \textit{contrary}, if the associated system of vector field flows~(\ref{T12}) \textit{is not compatible}, it is necessary to recover a hidden
isomorphism transformation
\begin{gather}
J^{1}\big(\mathbb{R}^{n+1}\times\mathbb{R}^{2};\mathbb{R}\big)\ni(x,w;\tau
)\rightarrow(x,\tilde{w};\tau)\in J^{1}\big(\mathbb{R}^{n+1}\times\mathbb{R}^{2};\mathbb{R}\big), \label{T14}
\end{gather}
for which the resulting \textit{a priori} \textit{compatible} first-order nonlinear differential relations
\begin{gather}
\frac{\partial\tilde{w}}{\partial\tau_{k}}+\sum_{j=1}^n \tilde {a}_{j}^{(k)}(x,\tilde{w};\tau)\frac{\partial\tilde{w}}{\partial x_{j}}-\tilde{b}^{(k)}(x,\tilde{w};\tau)=0 \label{T15}
\end{gather}
possess already \textit{a compatible associated system} of the corresponding vector field flows
\begin{gather}
\frac{\partial x_{j}}{\partial\tau_{k}}=\tilde{a}_{j}^{(k)}(x,\tilde{w};\tau) \label{T16}
\end{gather}
on the space $\mathbb{R}^{n}\times\mathbb{R}$, for which any solution $\tilde{w}\colon \mathbb{R}^{n+2}\rightarrow\mathbb{R}$ generates a first integral $\tilde{\psi}\colon \mathbb{R}^{n+1}\times\mathbb{R}^{2}\mathbb{\rightarrow R}$ of an associated compatible system of the linear vector field equations
\begin{gather}
\frac{\partial\tilde{\psi}}{\partial\tau_{k}}+\sum_{j=1}^n \tilde {a}_{j}^{(k)}(x,\lambda;\tau)\frac{\partial\tilde{\psi}}{\partial x_{j}}+\tilde{b}^{(k)}(x,\lambda;\tau)\frac{\partial\tilde{\psi}}{\partial\lambda }=0 \label{T17}
\end{gather}
on the space $\mathbb{R}^{n+1}\times\mathbb{R}^{2}\mathbb{\rightarrow R}$ for $k= 1,2$, where $\tilde{\psi}(x,\lambda;\tau):=\tilde{\alpha }(\tilde{w}(x;\tau)-\lambda)$ for all $(x,\lambda;\tau)\in\mathbb{R}^{n+1}\times\mathbb{R}$ and some smooth mapping $\tilde{\alpha}\colon \mathbb{R}\rightarrow\mathbb{R}$. From this one easily infers, as it was done above, a linearized covering jet submanifold $J_{\rm lin}^{1}\big(\mathbb{R}^{n+1}\times \mathbb{R}^{2};\mathbb{R}\big)|_{\mathcal{E}}\subset J^{1}\big(\mathbb{R}^{n+1}\times\mathbb{R}^{2}; \mathbb{R}\big)$, as a~compatible system of the vector field equations~(\ref{T17}), generated by the nonlinear first-order differential system $J^{1}\big(\mathbb{R}^{n}\times\mathbb{R}^{2}; \mathbb{R}\big)|_{\mathcal{E}}\subset J^{1}\big(\mathbb{R}^{n}\times\mathbb{R}^{2};\mathbb{R}\big)$ on the space $\mathbb{R}^{n}\times\mathbb{R}^{2}$. The latter exactly determines the inverse implication~(\ref{T10}) as applied to general compatible first-order nonlinear differential relations~(\ref{T13}), providing for the nonlinear first-order differential system $J^{1}\big(\mathbb{R}^{n}\times\mathbb{R}^{2};\mathbb{R}\big)|_{\mathcal{E}}\subset J^{1}\big(\mathbb{R}^{n}\times\mathbb{R}^{2};\mathbb{R}\big)$ its corresponding Lax--Sato representation.

\begin{Remark}
The existence of the mapping (\ref{T14}) can be inferred from the following reasoning. Assume that the mapping~(\ref{T14}) exists and is equivalent to the relation
\begin{gather}
w(x;\tau) :=\rho(x,\tilde{w}(x;\tau);\tau) \label{T17a}
\end{gather}
for some smooth function $\rho\colon \mathbb{R}^{n+1}\times\mathbb{R}^{2}\mathbb{\rightarrow R}$ and all $(x;\tau)\in\mathbb{R}^{n}\times \mathbb{R}^{2}$, where, by definition, the corresponding mapping $\tilde {w}\colon \mathbb{R}^{n}\times\mathbb{R}^{2}\rightarrow\mathbb{R}$ satisfies the following system of first-order differential relations:
\begin{gather}
\frac{\partial\tilde{w}}{\partial\tau_{k}}+\sum_{j=1}^n a_{j}%
^{(k)}(x,\rho(x,\tilde{w};\tau);\tau)\frac{\partial\tilde{w}}{\partial x_{j}%
}=\tilde{b}^{(k)}(x,\tilde{w};\tau), \label{T17b}
\end{gather}
compatible for all $\tau_{k}\in\mathbb{R}$, $k= 1,2$, and $x\in\mathbb{R}^{n}$. Here the functions $\tilde{b}^{(k)}\colon \mathbb{R}^{n+1}\times\mathbb{R}^{2}\mathbb{\rightarrow R}$, $k= 1,2 $, defined as
\begin{gather*}
\tilde{b}^{(k)}(x,\tilde{w};\tau):=\left[ b^{(k)}-\sum_{j=1}^n\left( \frac{\partial\rho}{\partial\tau_{k}}+a_{j}^{(k)}\frac{\partial\rho
}{\partial x_{j}}\right) \right] \left. \left( \frac{\partial\rho
}{\partial x_{j}}\right) ^{-1}\right\vert _{w =\rho(x,\tilde{w};\tau)},
\end{gather*}
should depend on the mapping (\ref{T17a}) in such a way that the vector fields
\begin{gather}
\frac{\partial x_{j}}{\partial\tau_{k}}=a_{j}^{(k)}(x,\rho(x,\tilde{w};\tau);\tau):=\tilde{a}_{j}^{(k)}(x,\tilde{w};\tau) \label{T17c}
\end{gather}
be also compatible for all $j=1,\dots,n $ and $k= 1,2 $ modulo the flows~(\ref{T17b}). The latter means that the equation~(\ref{T17b}) can be equivalently represented as a compatible system of the following vector field equations
\begin{gather}
\frac{\partial\tilde{\psi}}{\partial\tau_{k}}+\sum_{j=1}^n \tilde
{a}_{j}^{(k)}(x,\lambda;\tau)\frac{\partial\tilde{\psi}}{\partial x_{j}
}+\tilde{b}^{(k)}(x,\rho(x,\lambda;\tau);\tau)\frac{\partial\tilde{\psi}
}{\partial\lambda}=0 \label{T17d}
\end{gather}
on its first integral $\tilde{\psi}\colon \mathbb{R}^{n+1}\times\mathbb{R}^{2}\mathbb{\rightarrow R}$, where $\tilde{\psi}(x,\lambda;\tau)=\alpha
(\tilde{w}(x;\tau)-\lambda) $ for an arbitrarily chosen smooth mapping $\alpha\colon \mathbb{R\rightarrow R}$, any parameter $\lambda\in \mathbb{R}$ and all $(x;\tau)\in\mathbb{R}^{n}\times\mathbb{R}^{2}$. The resulting system~(\ref{T17d}) means exactly a suitable Lax--Sato type linearization of the compatible quasi-linear first-order differential relations~(\ref{T8}). Regarding the mapping~(\ref{T17a}) and the functions $\tilde{b}^{(k)}\colon \mathbb{R}^{n+1}\times \mathbb{R}^{2}\mathbb{\rightarrow R}$, $k= 1,2 $, that depend on it, one can easily observe that the compatibility condition for the vector fields~(\ref{T17c}) reduces to the a priori compatible first-order quasi-linear differential relations
\begin{gather}
\frac{\partial a_{j}^{(k)}}{\partial\rho}\frac{\partial\rho}{\partial
\tau_{s}}-\frac{\partial a_{j}^{(s)}}{\partial\rho}\frac{\partial\rho
}{\partial\tau_{k}}+\left( \frac{\partial a_{j}^{(k)}}{\partial\rho}
 \tilde{b}^{(s)}-\frac{\partial a_{j}^{(s)}}{\partial\rho}\tilde{b}
^{(k)}\right) \frac{\partial\rho}{\partial\tilde{w}}\nonumber\\
\qquad {}+\sum_{m=1}^n\left( \frac{\partial a_{j}^{(k)}}{\partial\rho}
a_{m}^{(s)}-\frac{\partial a_{j}^{(s)}}{\partial\rho}a_{m}^{(k)}\right)
\frac{\partial\rho}{\partial x_{m}}=0,
\label{T17e}
\end{gather}
where $j=1,\dots,n$ and $k\neq s= 1,2$, and whose solution is exactly the searched-for mapping~(\ref{T17a}). Taking here into account that we have only two functional parameters $b^{(s)}\colon \mathbb{R}^{n+1}\times\mathbb{R}^{2}\mathbb{\rightarrow R}$, $s=1,2$, the system of~$2n$ differential relations~(\ref{T17e}) can be, in general, compatible only for the case $n=1$. For all other cases $n\geq2$ the compatibility condition for~(\ref{T17e}) should be checked separately by direct calculations.
\end{Remark}

\subsection{\label{Subec_H4}\textbf{Example: the Gibbons--Tsarev equation }}

As a first degenerate case of the above scheme (\ref{T10}) we consider a compatible nonlinear first-order system $J^{1}\big(\mathbb{R}^{2};\mathbb{R}\big)|_{\mathcal{E}}$ $\subset J^{1}\big(\mathbb{R}^{2};\mathbb{R}\big)$ for $n=0$, discussed before in~\cite{Kras}:
\begin{gather}
\frac{\partial w}{\partial t}-\frac{1}{z_{y}+z_{t}w-w^{2}} =0,\qquad\frac{\partial w}{\partial y}+\frac{z_{t}-w}{u_{y}+z_{t}w-w^{2}}=0, \label{T18}
\end{gather}
where $(t,y;w)\in\mathbb{R}^{2}\times\mathbb{R}$ and a mapping $u\colon \mathbb{R}^{2}\rightarrow\mathbb{R}$ satisfies the Gibbons--Tsarev equation $\mathcal{E}[y,t;u]=0$ in the form
\begin{gather}
z_{yy}+z_{t}z_{ty}-z_{y}z_{tt}+1=0, \label{T18a}
\end{gather}
first derived in \cite{GiTs-1}. As the nonlinear system (\ref{T18}) is
compatible and the associated system of vector field flows (\ref{T16}) is empty, one easily infers \cite{OdSo} that any solution $w\colon \mathbb{R}^{2}\rightarrow\mathbb{R}$ to (\ref{T18}) generates a first
integral $\psi\colon \mathbb{R}\times\mathbb{R}^{2}\rightarrow\mathbb{R}$ of a~system of the vector field equations
\begin{gather}
\frac{\partial\psi}{\partial t}+\frac{1}{z_{y}+z_{t}\lambda-\lambda^{2}}\frac{\partial\psi}{\partial\lambda} =0,\qquad \frac{\partial\psi}{\partial y}-\frac{z_{t}-\lambda}{z_{y}+z_{t}\lambda-\lambda^{2}}\frac{\partial\psi }{\partial\lambda}=0, \label{T19}
\end{gather}
where, by definition, $\psi(\lambda;y,t):=\alpha(w(t,y)-\lambda)$ for all $(\lambda;t,y)\in\mathbb{R}\times\mathbb{R}^{2}$ and some smooth mapping $\alpha\colon \mathbb{R}\rightarrow\mathbb{R}$. The compatible system (\ref{T19}) considered as the jet submanifold $J_{\rm lin}^{1}\big(\mathbb{R}^{1}\times \mathbb{R}^{2};\mathbb{R}\big)|_{\mathcal{E}}\subset J^{1}\big(\mathbb{R}^{1}\times\mathbb{R}
^{2};\mathbb{R}\big)$ solves the problem of constructing the linearizing implication~(\ref{T10}).

As it was demonstrated in \cite{GiTs-2,Mokh}, the substitution
\begin{gather*}
u:=\frac{1}{2}\left(-z_{t}+\sqrt{z_{t}^{2}+4z_{y}}\right),\qquad v:=\frac{1}{2}\left(-z_{t}-\sqrt{z_{t}^{2}+4z_{y}}\right) 
\end{gather*}
gives rise to the next equivalent dynamical system
\begin{gather}
u_{y}=vu_{t}-(u-v)^{-1},\qquad v_{y}=uv_{t}+(u-v)^{-1} \label{T18c}
\end{gather}
on a functional space $M\subset C^{\infty}\big(\mathbb{R};\mathbb{R}^{2}\big)$ subject to the evolution parameter $y\in\mathbb{R}$ modulo evolution with respect to the joint evolution parameter $t\in\mathbb{R}$. Taking into account the Lax--Sato representation~(\ref{T19}), one easily infers~\cite{OdSo} the corresponding linearizing Lax--Sato representation for the
dynamical system (\ref{T18c}):
\begin{gather}
\frac{\partial\psi}{\partial t}-\frac{1}{(\lambda+u)(\lambda+v)}\frac {
\partial\psi}{\partial\lambda} =0,\qquad \frac{\partial\psi}{\partial y}-\frac{u+v+\lambda}{(\lambda+u)(\lambda+v)}\frac{\partial\psi}{\partial \lambda}=0.\label{T18d}
\end{gather}

The above Lax--Sato representation can be now reanalyzed more deeply within the Lie-algebraic scheme devised recently in \cite{HePrBlPr}. Namely,
we define the complex torus diffeomorphism Lie group $\tilde{G}:=\operatorname{Diff}\big(\mathbb{T}_{\mathbb{C}}^{1}\big)$, holomorphically extended to the interior $ \mathbb{D}_{+}^{1}\subset\mathbb{C}$ and to the exterior $\mathbb{D}_{-}^{1}\subset\mathbb{C}$ regions of the unit disc $\mathbb{D}^{1}\subset\mathbb{C}^{1}$, such that for any $g(\lambda)\in\tilde{G}|_{\mathbb{D}_{-}^{1}}$, $\lambda\in\mathbb{D}_{-}^{1}$, $g(\infty)=1\in \operatorname{Diff}\big(\mathbb{T}^{1}\big)$, and study its specially chosen coadjoint orbits related to the compatible system of linear vector field equations~(\ref{T18d}).

As a first step for solving this problem one needs to consider the corresponding Lie algebra $\mathcal{\tilde{G}}:=\operatorname{Diff}\big(\mathbb{T}_{\mathbb{C}}^{1}\big)$ and its decomposition into the direct sum of subalgebras
\begin{gather*}
\mathcal{\tilde{G}}=\mathcal{\tilde{G}}_{+}\oplus\mathcal{\tilde{G}}_{-} 
\end{gather*}
of Laurent series with positive as $|\lambda|\rightarrow0$ and strongly negative as $\lambda|\rightarrow\infty$ degrees, respectively. Then, owing
to the classical Adler--Kostant--Symes theory, for any element $\tilde{l}\in \mathcal{\tilde{G}}^{\ast}\simeq\Lambda^{1}\big(\mathbb{T}_{\mathbb{C}}^{1}\big)$ the following formally constructed flows
\begin{gather}
{\rm d}\tilde{l}/{\rm d}y=-\operatorname{ad}_{\nabla h^{(y)}(\tilde{l})}^{\ast}\tilde{l},
\qquad {\rm d}\tilde{l}/{\rm d}t=-\operatorname{ad}_{\nabla h^{(t)}(\tilde{l})}^{\ast}\tilde{l} \label{T20d}
\end{gather}
with the evolution parameters $y,t\in\mathbb{R}^{2}$ are always compatible, if $h^{(p_{y})}$ and $h^{(p_{t})}\in I\big(\mathcal{\tilde{G}}^{\ast}\big)$ are arbitrarily chosen functionally independent Casimir functionals on the adjoint space $\mathcal{\tilde{G}}^{\ast}$, and by definition, $\nabla h^{(y)}(\tilde{l}):=\nabla h^{(p_{y})}(\tilde{l})_{-}$, $\nabla h^{(t)}(\tilde {l}):=\nabla h^{(p_{t})}(\tilde{l})_{-}$ are their gradients, suitably projected onto the subalgebra $\mathcal{\tilde{G}}_{-}$. Bearing in mind the above classical result, consider the Casimir functional $h^{(p_{y})}$ on $\mathcal{\tilde{G}}^{\ast}$, whose gradient $\nabla h^{(p_{y})}(\tilde{l}):=\nabla h^{(p_{y})}(l)\partial/\partial\lambda \mathcal{\in \tilde{G}}$ as $|\lambda|\rightarrow\infty$ is taken, for simplicity, in the asymptotic form
\begin{gather*}
\nabla h^{(p_{y})}(\tilde{l})\sim\left( \frac{\lambda+u+v}{(\lambda+u)(\lambda+v)}+\alpha_{0}+\alpha_{1}\lambda\right) \frac{\partial}{\partial\lambda}, 
\end{gather*}
where $\lambda\in\mathbb{T}_{\mathbb{C}}^{1}$, $|\lambda|\rightarrow\infty$, and coefficients $\alpha_{j} \in C^{\infty}\big(\mathbb{R}^{2};\mathbb{R}\big)$, $j=0,1$, are arbitrarily chosen nontrivial functional parameters, giving rise to the gradient projection
\begin{gather}
\nabla h^{(y)}(\tilde{l}):=\nabla h^{(p_{y})}(\tilde{l})|_{-}=\frac {\lambda+u+v}{(\lambda+u)(\lambda+v)}, \label{T22d}
\end{gather}
generating the first flow of (\ref{T20d}). As the differential 1-form $\tilde{l}=l(\lambda;y,t){\rm d}\lambda \in\Lambda^{1}\big(\mathbb{T}_{\mathbb{C}}^{1}\big)\simeq\mathcal{\tilde{G}}$ satisfies, by definition, the condition
\begin{gather}
\operatorname{ad}_{\nabla h^{(p_{y})}(\tilde{l})}^{\ast}\tilde{l}=0, \label{T22dd}
\end{gather}
equivalent to the differential equation
\begin{gather}
\frac{{\rm d}}{{\rm d}\lambda}\big[l(\lambda;y,t)\big(\nabla h^{(p_{y})}(l)\big)^{2}\big]=0,\label{T23d}
\end{gather}
one easily infers from (\ref{T22d}) and (\ref{T23d}) that the coefficient
\begin{gather*}
l(\lambda;y,t)= \big(\nabla h^{(p_{y})}(l)\big)^{-2}= \frac{(\lambda+u)^{2}(\lambda+v)^{2}}{[\lambda+u+v+(\alpha_{0}+\alpha_{1}\lambda)(\lambda+u)(\lambda+v)]^{2}},
\end{gather*}
satisfies the relation $l(\lambda;y,t)\big(\nabla h^{(p_{y})}(l)\big)^{2}=1$ for all $(t,y)\in\mathbb{R}^{2}$.

Now we formulate the following useful observation.

\begin{Lemma} \label{Lm_T1} The set $\mathrm{I}(\mathcal{\tilde{G}}^{\ast})$ of the functionally independent Casimir invariants is one-dimen\-sional.
\end{Lemma}

\begin{proof}Any asymptotic solution to the determining equation (\ref{T26d}) is invariant under the multiplication $\nabla h^{(p_{t})}(\tilde{l}) \rightarrow\sigma(t,y)\nabla h^{(p_{t})}(\tilde{l})$ by an arbitrary smooth function $\sigma\colon \mathbb{R}^{2}\rightarrow\mathbb{R}$. This proves the lemma.
\end{proof}

Consider now the gradient $\nabla h^{(p_{t})}(\tilde{l})\in\mathcal{\tilde{G}}$ of the Casimir functional $h^{(p_{t})}\in\mathrm{D} (\mathcal{\tilde{G}}^{\ast}\mathcal{)}$, that satisfies the condition, which is identical to (\ref{T22dd}), that is,
\begin{gather*}
\operatorname{ad}_{\nabla h^{(p_{t})}(\tilde{l})}^{\ast}\tilde{l}=0, 
\end{gather*}
which is equivalent to the following linear differential equation
\begin{gather}
2l(\lambda;y,t)\frac{\partial}{\partial\lambda}\nabla h^{(p_{t})}(\tilde {l})+\nabla h^{(p_{t})}(\tilde{l})\frac{\partial}{\partial\lambda}l(\lambda;y,t)=0, \label{T26d}
\end{gather}
and whose solution can be naturally represented as the asymptotic series{\samepage
\begin{gather}
\nabla h^{(p_{t})}(\tilde{l})\sim\frac{1}{(\lambda+u)(\lambda+v)}+\sum _{j\in\mathbb{Z}_{+}}\beta_{j}\lambda^{j} \label{T27d}
\end{gather}
for some nontrivial coefficients $\beta_{j}\in C^{\infty}\big(\mathbb{R}^{2};\mathbb{R}\big) $ successively determined from the equation~(\ref{T26d}).}

As follows from Lemma \ref{Lm_T1}, the asymptotic solution~(\ref{T27d}) to the determining equa\-tion~(\ref{T26d}) exists exclusively owing to its natural invariance under the multiplication
\begin{gather*}
\nabla h^{(p_{t})}(\tilde{l})\rightarrow\sigma(t,y)\nabla h^{(p_{t})}(\tilde{l})
 \end{gather*} by arbitrary smooth function $\sigma\colon \mathbb{R}^{2}\rightarrow\mathbb{R}$. In particular, this means that asymptotically as $\lambda\rightarrow \infty$ the product $l(\lambda;y,t)(\nabla h^{(p_{t})}(l))^{2}\sim0$, for otherwise, if $l(\lambda;y,t)(\nabla h^{(p_{t})}(l))^{2}\nrightarrow0$, this product becomes, owing to~(\ref{T23d}), a nonzero constant involving the parameter $\lambda\in\mathbb{C}$. The latter, evidently, means that $\nabla h^{(p_{t})}(l) =\nabla h^{(p_{y})}(\tilde {l})\sigma(t,y)$ for any smooth arbitrary function $\sigma \in C^{\infty }\big(\mathbb{R}^{2};\mathbb{R}\big)$, producing no new flow with respect to the evolution parameter $t\in \mathbb{R}$.

Construct now the gradient projection
\begin{gather*}
\nabla h^{(t)}(\tilde{l}):=\nabla h^{(p_{t})}(\tilde{l})|_{-}=\frac {1}{(\lambda+u)(\lambda+v)}
\end{gather*}
generating the second flow of (\ref{T20d}). As a consequence of the above results we can easily derive the following compatibility condition for the flows~(\ref{T20d}):
\begin{gather*}
\left[\frac{\partial}{\partial y}-\nabla h^{(y)}(\tilde{l}),\frac{\partial }{\partial y}-\nabla h^{(ty)}(\tilde{l})\right]=0, 
\end{gather*}
which is equivalent modulo the dynamical system~(\ref{T18c}) to the following system of two \textit{a priori }compatible linear vector field equations:
\begin{gather}
\frac{\partial\psi}{\partial y}-\frac{\lambda+u+v}{(\lambda+u)(\lambda +v)}
\frac{\partial\psi}{\partial\lambda}=0,\qquad \frac{\partial\psi }{\partial t}-\frac{1}{(\lambda+u)(\lambda+v)} \frac{\partial\psi}{\partial\lambda}=0 \label{T30d}
\end{gather}
for $\psi\in C^{2}(\mathbb{R};\mathbb{R})$ and all $(\lambda;t,y)\in \mathbb{C\times R}^{2}$. Thus, we can formulate the results, obtained above, as the following proposition.

\begin{Proposition}\label{Prop_T3}A system $J_{\rm lin}^{1}\big(\mathbb{R}\times\mathbb{R}^{2};\mathbb{R}\big)|_{\mathcal{E}}$ of the linear first-order differential relations~\eqref{T30d} on the covering jet manifold $J^{1}\big(\mathbb{R}\times\mathbb{R}^{2};\mathbb{R}\big)$ is compatible, that is, it holds for any its smooth solution $\psi\colon \mathbb{R}\times\mathbb{R} ^{2}\rightarrow\mathbb{R}$ at all points $(\lambda;y,t)\in \mathbb{R} \times\mathbb{R}^{2} $ iff the function $u\colon \mathbb{R}^{2}\rightarrow \mathbb{R}$ satisfies the Gibbons--Tsarev equation~\eqref{T18a}.
\end{Proposition}

Moreover, taking into account that the flows (\ref{T20d}) are, by construction, Hamiltonian systems on the coadjoint space $\mathcal{\tilde {G}}^{\ast}$, their reduction on the space of functional variables $(u,v)\in C^{\infty}\big(\mathbb{R}^{2};\mathbb{R}^{2}\big)$ will be, respectively, Hamiltonian too. This reduction scheme is now under study and is expected to be presented elsewhere.

\subsection{Example: the ABC equation}

For the second example we consider a compatible system $J^{1}\big(\mathbb{R}\times \mathbb{R}^{2};\mathbb{R}\big)|_{\mathcal{E}}$ of the nonlinear first-order differential relations~(\ref{T3}) on the jet manifold $J^{1}\big(\mathbb{R}\times\mathbb{R}^{2};\mathbb{R}\big)$. It is easy to check that the associated system of vector field flows
\begin{gather*}
\frac{\partial x}{\partial t}=\frac{u_{t}w}{u_{x}k_{1}(k_{1}+k_{2}-1)},\qquad \frac{\partial x}{\partial y}=\frac{u_{y}w}{u_{x}k_{1}(w+k_{1}+k_{2}-1)}, 
\end{gather*}
modulo differential relations (\ref{T3}), is not compatible for all $(w;t,y)\in \mathbb{R}\times\mathbb{R}^{2}$. Thus, we need to construct
such a mapping~(\ref{T14}) that the resulting system of nonlinear differential relations~(\ref{T15}) will possess an associated system of vector field flows already compatible for all $(\tilde{w};t,y)\in \mathbb{R}\times\mathbb{R}^{2}$. To do this let us get rid from the very
beginning in the equations~(\ref{T3}) of the strictly linear part, giving rise to the representation its solution as $w(x;t,y)=(u_{x}(x;t,y))^{\alpha} \tilde {w}(x;t,y)$ for $\alpha=(k_{1}+k_{2}-1)/k_{1}$ and all $(x;t,y)\in \mathbb{R}\times\mathbb{R}^{2}$. Substituting this into (\ref{T3}) yields the following \textit{a~priori} compatible system $J^{1}\big(\mathbb{R}\times\mathbb{R}^{2};\mathbb{R}\big)|_{\mathcal{E}}$ of the first-order differential relations, cf.~\cite{KrSeMo},
\begin{gather}
\frac{\partial\tilde{w}}{\partial t}+ \frac{\tilde{w}u_{x}^{\alpha-1}u_{t}}{k_{1}\alpha}\frac{\partial\tilde{w}}{\partial x}+\frac{\tilde{w} ^{2}u_{x}^{\alpha}(u_{t}u_{xx}-k_{1}u_{xx}u_{xt})}{k_{1}^{2}}=0,\nonumber \\
\frac{\partial\tilde{w}}{\partial y}+ \frac{k_{2}\tilde{w}u_{x}^{\alpha-1}u_{y}}{k_{1}(\tilde{w}u_{x}^{\alpha}+\alpha)}\frac{\partial\tilde{w}}{\partial x}+\frac{k_{2}\alpha\tilde{w}^{2}u_{x}^{\alpha}u_{y}u_{xx}}{k_{1}^{2}(\tilde{w}u_{x}^{\alpha}+\alpha)}=0\label{T22}
\end{gather}
on the jet manifold $J^{1}\big(\mathbb{R}\times\mathbb{R}^{2};\mathbb{R}\big)$. We can now easily check that the above expression $\tilde {w}(x;t,y)=(u_{x}(x;t,y))^{-\alpha}w(x;t,y)$ for all $(x;t,y)\in \mathbb{R}\times\mathbb{R}^{2}$ determining the mapping~(\ref{T14}), is exactly the one we searched for, inasmuch as the associated system of vector field flows
\begin{gather*}
\frac{\partial x}{\partial t}= \frac{\tilde{w}u_{x}^{\alpha-1}u_{t}}{k_{1}\alpha},\qquad \frac{\partial x}{\partial y}= \frac{k_{2}\tilde {w}u_{x}^{\alpha-1}u_{y}}{k_{1}(\tilde{w}u_{x}^{\alpha}+\alpha)} 
\end{gather*}
on $\mathbb{R}\times\mathbb{R}^{2}$ proves to be compatible modulo the system $J^{1}\big(\mathbb{R}\times\mathbb{R}^{2};\mathbb{R}\big)|_{\mathcal{E}}$ of \textit{a priori} compatible differential relations (\ref{T22}) on $J^{1}\big(\mathbb{R}\times\mathbb{R}^{2};\mathbb{R}\big)$. Based on this compatibility result one can easily construct the corresponding linearizing system $J_{\rm lin}^{1}\big(\mathbb{R}^{2} \times\mathbb{R}^{2};\mathbb{R}\big)|_{\mathcal{E}}$ on the covering jet manifold $J^{1}\big(\mathbb{R}^{2}\times\mathbb{R}^{2};\mathbb{R}\big)$ by implemening the inverse implication (\ref{T10}) as the Lax--Sato representation
\begin{gather}
\frac{\partial\psi}{\partial t}+ \frac{\lambda u_{x}^{\alpha-1}u_{t}}{k_{1}\alpha}\frac{\partial\psi}{\partial x}-\frac{\lambda^{2}u_{x}^{\alpha
}(u_{t}u_{xx}-k_{1}u_{xx}u_{xt})}{k_{1}^{2} }\frac{\partial\psi}{\partial\lambda}=0, \nonumber\\
\frac{\partial\psi}{\partial y}+ \frac{k_{2}\lambda u_{x}^{\alpha-1}u_{y}}{k_{1}(\lambda u_{x}^{\alpha}+\alpha)}\frac{\partial\psi}{\partial x}-\frac{\lambda^{2}k_{2}\alpha u_{x}^{\alpha-1}u_{y}u_{xx}}{k_{1}^{2}(\lambda
u_{x}^{\alpha}+\alpha)}\frac{\partial\psi}{\partial\lambda}=0,\label{T23}
\end{gather}
exactly coinciding with (\ref{T4}), where, by definition, $\psi (x,\lambda;t,y):=\alpha(\tilde{w}(x;t,y)-\lambda)$ for all $(x,\lambda ;t,y)\in\mathbb{R}^{2}\times\mathbb{R}^{2}$ and any smooth mapping $\alpha\colon \mathbb{R}\rightarrow\mathbb{R}$. Thus, the linear differential system (\ref{T23}) solves the problem, which was posed above, of constructing the inverse implication~(\ref{T10}) for the compatible nonlinear differential system $J^{1}\big(\mathbb{R}\times\mathbb{R}^{2};\mathbb{R}\big)|_{\mathcal{E}}$~(\ref{T3}), thereby proving Proposition~\ref{T1}.

\subsection{Example: the Manakov--Santini equations}

The Manakov--Santini equations
\begin{gather}
u_{tx}+u_{yy}+(uu_{x})_{x}+v_{x}u_{xy}-v_{y}u_{xx} =0,\nonumber \\
v_{xt}+v_{yy}+uv_{xx}+v_{x}v_{xy}-v_{y}v_{xx} =0, \label{T24}
\end{gather}
where $(u,v)\in C^{\infty}\big(\mathbb{R}\times\mathbb{R}^{2};\mathbb{R}^{2}\big)$, as is well known \cite{MaSa}, are obtained as some generalization of the dispersionless limit for the Kadomtsev--Petviashvili equation. Equations~(\ref{T24}) possess the following compatible nonlinear first-order differential covering $J^{1}\big(\mathbb{R}\times\mathbb{R}^{2};\mathbb{R}\big)|_{\mathcal{E}}\subset J^{1}\big(\mathbb{R}\times\mathbb{R}^{2};\mathbb{R}\big)$:
\begin{gather}
\frac{\partial w}{\partial t}+\big(w^{2}-wv_{x}+u-v_{y}\big)\frac{\partial w}{\partial x}+(w-v_{x})(u_{x}-wv_{xx})-u_{y}+v_{yy}+v_{x}v_{xy}=0, \nonumber\\
\frac{\partial w}{\partial y}+w\frac{\partial w}{\partial x}-v_{xx}w+(u-v_{y})_{x}=0,\label{T25}
\end{gather}
giving rise for all $(x;t,y)\in\mathbb{R}\times\mathbb{R}^{2}$ to the Manakov--Santini differential relation $\mathcal{E}[x;y,t;u,v]=0$ (\ref{T24}) as a~submanifold on the associated jet manifold $J^{2}\big(\mathbb{R}\times\mathbb{R}^{2};\mathbb{R}^{2}\big)$. It is now easy to check that the system of vector field flows
\begin{gather*}
\frac{\partial x}{\partial t}=w^{2}-wv_{x}+u-v_{y},\qquad\frac{\partial x}{\partial y}=w, 
\end{gather*}
naturally related to~(\ref{T25}), is for all $(t,y)\in\mathbb{R}^{2}$ not compatible modulo these differential relations~(\ref{T25}). Thus, one needs to construct such an isomorphic transformation
\begin{gather*}
J^{1}\big(\mathbb{R}^{2}\times\mathbb{R}^{2};\mathbb{R}\big)\ni(x,w;t,y)\rightarrow
(x,\tilde{w};t,y)\in J^{1}\big(\mathbb{R}^{2}\times\mathbb{R}^{2};\mathbb{R}\big)
\end{gather*}
that the expression $w(x;y,t)=\beta(x,\tilde{w};y,t)$ for some, in general nonlinear, smooth mapping $\beta\colon \mathbb{R}^{2}\times \mathbb{R}^{2}\rightarrow\mathbb{R}$ transforms the first-order differential covering~(\ref{T25}) into an equivalent compatible first-order differential covering~(\ref{T15}), for which the corresponding vector field flows~(\ref{T16}) already become compatible. This problem is solved easily enough giving rise to the simple mapping
\begin{gather*}
w=\tilde{w}+v_{x} 
\end{gather*}
using which one obtains the following compatible first-order differential covering:
\begin{gather}
\frac{\partial\tilde{w}}{\partial t}+\big(\tilde{w}^{2}+\tilde{w}v_{x}+u-v_{y}\big)
\frac{\partial\tilde{w}}{\partial x}-u_{y}+\tilde{w}u_{x}=0, \nonumber\\
\frac{\partial\tilde{w}}{\partial y}+(v_{x}+\tilde{w})\frac{\partial\tilde{w}}{\partial x}+u_{x}=0.\label{T29}
\end{gather}
It is easy to check that the system of vector field flows
\begin{gather*}
\frac{\partial x}{\partial t}=\tilde{w}^{2}+\tilde{w}v_{x}+u-v_{y}, \qquad\frac{\partial x}{\partial y}=v_{x}+\tilde{w}, 
\end{gather*}
naturally related to (\ref{T29}), is already compatible for all $(t,y)\in\mathbb{R}^{2}$ modulo the differential relations~(\ref{T29}). Based on this compatibility result stated above, one can easily construct the associated linearizing first-order differential system $J_{\rm lin}^{1}\big(\mathbb{R}^{2}\times\mathbb{R} ^{2};\mathbb{R}\big)|_{\mathcal{E}}$ on the covering jet manifold $J^{1}\big(\mathbb{R}^{2}\times\mathbb{R}^{2};\mathbb{R}\big)$, realizing the inverse
implication (\ref{T10}) as the Lax--Sato representation
\begin{gather*}
\frac{\partial\psi}{\partial t}+\big(\lambda^{2}+\lambda v_{x}+u-v_{y}\big)\frac{\partial\psi}{\partial x}+(u_{y}-\lambda u_{x})\frac{\partial\psi }{\partial\lambda}=0, \nonumber\\
\frac{\partial\psi}{\partial y}+(v_{x}+\lambda)\frac{\partial\psi}{\partial x }-u_{x}\frac{\partial\psi}{\partial\lambda}=0,
\end{gather*}
thus proving Proposition \ref{Prop_T2}.

\section{The contact geometry linearization covering scheme}\label{Sec_T3}

\subsection{The setup}

We consider two Hamilton--Jacobi-type compatible for all $(x;\tau):=(x;t,y)\in\mathbb{R}\times\mathbb{R}^{2}$ in general nonlinear first-order differential relations:
\begin{gather}
\frac{\partial z}{\partial t}+\tilde{H}^{(t)}(x,z,\partial z/\partial
x;t,y)=0,\qquad \frac{\partial z}{\partial y}+\tilde{H}^{(y)}(x,z,\partial z/\partial x;t,y)=0, \label{Th1}
\end{gather}
where $z:\mathbb{R}^{3}\mathbb{\rightarrow R}$ is the so-called ``\textit{action function}'' and $\tilde{H}^{(t)},\tilde{H}^{(y)}\colon \mathbb{R}^{3}\times \mathbb{R}^{2}\mathbb{\rightarrow R}$ are some smooth generalized Hamiltonian functions. The relations (\ref{Th1}) follow from the contact geometry \cite{BrGrGr,DuNoFo} interpretation of some mechanical systems generated by some vector fields. Namely, a differential one-form $\alpha ^{(1)}$ $\in\Lambda^{1}\big(\mathbb{R}^{3}\times\mathbb{R}\big)$, defined by the expression
\begin{gather*}
\alpha^{(1)}:={\rm d}z-\lambda {\rm d}x, 
\end{gather*}
is \textit{contact}, as $\alpha^{(1)}\wedge d\alpha^{(1)}$ defines a volume form on~$\mathbb{R}^{3}$, and vector fields $X_{H^{(t)}},X_{H^{(y)}}\in\Gamma\big(T\big(\mathbb{R}^{3}\times\mathbb{R}\big)\big)$ are called, respectively, \textit{contact vector fields} (with respect to $\alpha^{(1)}$), if there exist some functions $\mu^{(t)},\mu^{(y)}\colon \mathbb{R}^{3}\times\mathbb{R}\rightarrow\mathbb{R}$ such that for all $(x,z;\tau)\in\mathbb{R}^{2}\times\mathbb{R}^{2}$ the
following equalities
\begin{gather}
-i_{X_{H^{(t)}}}\alpha^{(1)} =H^{(t)}:=\tilde{H}^{(t)}|_{\partial
z/\partial x=\lambda},\qquad -i_{X_{H^{(y)}}}\alpha^{(1)}=H^{(y)}:=\tilde{H}^{(y)}|_{\partial z/\partial x=\lambda},\nonumber \\
\mathcal{L}_{H^{(t)}}\alpha^{(t)} =\mu^{(t)}\alpha^{(t)},\qquad \mathcal{L}_{H^{(y)}}\alpha^{(y)}=\mu ^{(y)}\alpha^{(y)}\label{Th3}
\end{gather}
hold, where $\mathcal{L}_{H^{(t)}}$, $\mathcal{L}_{H^{(y)}}$ are the Lie derivatives~\cite{Cart,DuNoFo,Godb} with respect to the
vector fields $X_{H^{(t)}},X_{H^{(y)}}\in\Gamma\big( T\big(\mathbb{R}^{3}\times \mathbb{R}\big)\big)$. From the conditions~(\ref{Th3}) one easily finds~\cite{BlPrSaWaKu, Kiri} that
\begin{gather}
X_{H^{(t)}} =\frac{\partial H^{(t)}}{\partial\lambda}\frac{\partial }{\partial x}-\left(\frac{\partial H^{(t)}}{\partial x}+\lambda\frac{\partial
H^{(t)}}{\partial z}\right)\frac{\partial}{\partial\lambda}+\left(-H^{(t)}+\lambda
\frac{\partial H^{(t)}}{\partial\lambda}\right)\frac{\partial}{\partial z}, \nonumber\\
X_{H^{(y)}} =\frac{\partial H^{(y)}}{\partial\lambda}\frac{\partial }{\partial x}-\left(\frac{\partial H^{(y)}}{\partial x}+\lambda\frac{\partial
H^{(y)}}{\partial z}\right)\frac{\partial}{\partial\lambda}+\left(-H^{(y)}+\lambda
\frac{\partial H^{(y)}}{\partial\lambda}\right)\frac{\partial}{\partial z},\label{Th4}
\end{gather}
where we have put $H^{(t)}:=\tilde{H}^{(t)}|_{\partial z/\partial x=\lambda }$, $H^{(y)}:=\tilde{H}^{(y)}|_{\partial z/\partial x=\lambda}$, and compatibility of the nonlinear relations~(\ref{Th1}) is equivalent to the commutativity condition for the following vector fields:
\begin{gather*}
\left[\frac{\partial}{\partial t}+X_{H^{(t)}},\frac{\partial}{\partial y}+X_{H^{(y)}}\right]=0 
\end{gather*}
for all $(x,z,\lambda;\tau)\in\mathbb{R}^{3}\times\mathbb{R}^{2}$, depending parametrically on $\lambda\in\mathbb{R}$. The latter condition can be rewritten as a~compatible Lax--Sato representation for the vector field equations
\begin{gather*}
\frac{\partial\psi}{\partial t}+X_{H^{(t)}}\psi=0,\qquad \frac{\partial \psi}{\partial y}+X_{H^{(y)}}\psi=0 
\end{gather*}
for smooth invariant functions $\psi\in C^{2}\big(\mathbb{R}^{3}\times \mathbb{R}^{2};\mathbb{R}\big)$ and all $(x,z,\lambda;\tau)\in\mathbb{R}^{3}\times\mathbb{R}^{2}$.

\begin{Remark} It is worth mentioning here that in the case when the Hamiltonian functions in~(\ref{Th1}) do not depend on the ``\textit{action function}'' $z\colon \mathbb{R}^{3}\mathbb{\rightarrow R}$, the contact vector fields naturally reduce to the classical Hamiltonian ones:
\begin{gather*}
X_{H^{(t)}}=\frac{\partial H^{(t)}}{\partial\lambda}\frac{\partial}{\partial
x}-\frac{\partial H^{(t)}}{\partial x}\frac{\partial}{\partial\lambda},\qquad X_{H^{(y)}}=\frac{\partial H^{(y)}}{\partial\lambda}\frac{\partial
}{\partial x}-\frac{\partial H^{(y)}}{\partial x}\frac{\partial}{\partial\lambda}, 
\end{gather*}
well known \cite{DuNoFo} from the symplectic geometry.
\end{Remark}

\subsection{Example: the differential Toda singular manifold equation}

The above general contact geometry linearization scheme is a modification of the one originally discovered and applied to integrable 4D disper\-sion\-less equations by A.~Sergyeyev in~\cite{Serg-1}. Below we will apply this scheme to a~degenerate case when the system (\ref{Th1}) is given by the following linear first-order differential relations:
\begin{gather}
\frac{\partial z}{\partial t}+\frac{\big(e^{-2u_{t}}-1\big)}{2u_{x}}\frac{\partial
z}{\partial x}=0,\qquad \frac{\partial z}{\partial y}-u_{y}u_{x}^{-1}e^{-2u_{t}}\frac{\partial z}{\partial x}=0 \label{Th7}
\end{gather}
for a smooth mapping $z\colon \mathbb{R}\times\mathbb{R}^{2}\mathbb{\rightarrow R}$, whose compatibility condition is the interesting~\cite{BuFeTs} differential Toda singular manifold equation~(\ref{T4ab}) for a smooth function $u\colon \mathbb{R}\times\mathbb{R}^{2}\mathbb{\rightarrow R}$ for all $(x;y,t)\in\mathbb{R}\times\mathbb{R}^{2}$, defining a differential relation $J^{2}\big(\mathbb{R}^{2}\times\mathbb{R}^{2};\mathbb{R}\big)|_{\mathcal{E}} \subset J^{2}\big(\mathbb{R}^{2}\times\mathbb{R}^{2};\mathbb{R}\big)$
on the jet manifold $J^{2}\big(\mathbb{R}^{2}\times \mathbb{R}^{2};\mathbb{R}\big)$
\begin{gather}
u_{xy}\sinh^{2}u_{t}=u_{x}u_{y}u_{tt}. \label{Th8}
\end{gather}
Although the differential relations (\ref{Th7}) are linear, they contain no ``spectral'' parameter $\lambda\in\mathbb{R}$ using which
one can construct the conservation laws for (\ref{T8}) and apply the modified inverse scattering transform for constructing its particular exact solutions.

However, the above contact geometry linearization scheme makes it possible to represent the system in question as a set of the compatible Hamilton--Jacobi equations
\begin{gather*}
\frac{\partial z}{\partial t}+\frac{\big(e^{-2u_{t}}-1\big)}{2u_{x}}\lambda=0,\qquad
\frac{\partial z}{\partial y}-\frac{u_{y}e^{-2u_{t}}}{u_{x}}\lambda =0
\end{gather*}
with the contact Hamiltonians
\begin{gather*}
H^{(t)}:=\frac{\big(e^{-2u_{t}}-1\big)}{2u_{x}}\lambda,\qquad H^{(y)}:=-\frac{u_{y}e^{-2u_{t}}}{u_{x}}\lambda, 
\end{gather*}
where the function $u\colon \mathbb{R}^{2}\times\mathbb{R}^{2}\mathbb{\rightarrow R}$ depends here on the additional variable $z\in\mathbb{R}$. Taking into account (\ref{Th4}), one can construct the corresponding extended contact vector fields
\begin{gather*}
\tilde{X}_{H^{(t)}} :=\frac{\partial}{\partial t}+X_{H^{(t)}}=\frac{\partial}{\partial t}+\frac{(e^{-2u_{t}}-1)}{2u_{x}}\frac{\partial }{\partial x}-\left[ \lambda\left( \frac{e^{-2u_{t}}-1}{2u_{x}}\right)_{x}+\lambda^{2}\left( \frac{e^{-2u_{t}}-1}{2u_{x}}\right) _{z}\right] \frac{\partial}{\partial\lambda}, \nonumber\\
\tilde{X}_{H^{(y)}} :=\frac{\partial}{\partial y}+X_{H^{(y)}}=\frac{\partial}{\partial y}-\frac{u_{y}e^{-2u_{t}}}{u_{x}}\frac{\partial }{\partial x}+\left[ \lambda\left( \frac{u_{y}e^{-2u_{t}}}{u_{x}}\right)_{x}+\lambda^{2}\left( \frac{u_{y}e^{-2u_{t}}}{u_{x}}\right) _{z}\right]
\frac{\partial}{\partial\lambda},
\end{gather*}
compatible on the solution set to the nonlinear differential relation (\ref{Th8}) on the jet manifold $J^{2}\big(\mathbb{R}^{2}\times\mathbb{R}^{2};\mathbb{R}\big)$. The latter makes it possible to state our final proposition on the contact geometry linearization covering of the system~(\ref{Th8}).

\begin{Proposition}The linear first-order differential relations~\eqref{Th7} on the jet manifold $J^{2}\big(\mathbb{R}^{2}\times\mathbb{R}^{2};\mathbb{R}\big)$ allow the following quadratic in the parameter $\lambda\in\mathbb{R}$ Lax--Sato type linearization covering
\begin{gather*}
\frac{\partial\psi}{\partial t}+\frac{\big(e^{-2u_{t}}-1\big)}{2u_{x}}\frac
{\partial\psi}{\partial x}-\left[ \lambda\left( \frac{e^{-2u_{t}}-1}{2u_{x}}\right) _{x}+\lambda^{2}\left( \frac{e^{-2u_{t}}-1}{2u_{x}}\right)
_{z}\right] \frac{\partial\psi}{\partial\lambda} =0,\nonumber\\
\frac{\partial\psi}{\partial y}-\frac{u_{y}e^{-2u_{t}}}{u_{x}}\frac
{\partial\psi}{\partial x}+\left[ \lambda\left( \frac{u_{y}e^{-2u_{t}}}{u_{x}}\right) _{x}+\lambda^{2}\left( \frac{u_{y}e^{-2u_{t}}}{u_{x}}\right) _{z}\right] \frac{\partial\psi}{\partial\lambda} =0
\end{gather*}
for smooth invariant functions $\psi\in C^{2}\big(\mathbb{R}^{3}\times\mathbb{R}^{2};\mathbb{R}\big)$, all $(x,z,\lambda;\tau)\in\mathbb{R}^{3}\times \mathbb{R}^{2}$ and any smooth solution $u\colon \mathbb{R}^{2}\times\mathbb{R}^{2}\mathbb{\rightarrow R}$ to the differential relation~\eqref{Th8}.
\end{Proposition}

\section{Conclusions}

We have studied an interesting problem posed by I.~Krasil'shchik in \cite{Kras} concerning a differential-geometric explanation of the linearization
procedure for a given nonlinear differential-geometric relation $J^{1}\big(\mathbb{R}^{n}\times\mathbb{R}^{2};\mathbb{R}\big)|_{\mathcal{E}}$ in the
jet manifold $J^{1}\big(\mathbb{R}^{n}\times \mathbb{R}^{2};\mathbb{R}\big)$, $n\in \mathbb{Z}_{+}$, realizing a compatible \textit{covering} \cite{BaKrMoVo,OdSo} for the corresponding nonlinear differential equation $\mathcal{E}[x,\tau;u]=0$, imbedded into some associated jet manifold $J^{k}\big(\mathbb{R}^{n}\times\mathbb{R}^{2};\mathbb{R}^{m}\big)$ for some $k,m\in\mathbb{Z}_{+}$. We have analyzed this covering jet manifold scheme, its relation to the invariant theory of the quasi-linear associated vector fields and applications to the Lax--Sato type integrability of nonlinear dispersionless differential systems. In particular, we discussed the compatible Jacobi type nonlinear relations, the related contact geometry and its application to the linearization covering scheme. The devised techniques were applied to such nonlinear Lax--Sato integrable equations as Gibbons--Tsarev, ABC, Manakov--Santini and the differential Toda singular manifold equations.

\subsection*{Acknowledgements}

The author cordially thanks Professors M.~B{\l}aszak, B.~Szablikowski and J.~Cie\'{s}linski for useful discussions of the results during the
International Conference in Functional Analysis dedicated to the 125$^{\rm th}$ anniversary of Stefan Banach held on 18--23 September, 2017 in Lviv,
Ukraine. He is also greatly indebted to Professors V.E.~Zakharov (University of Arizona, Tucson) and J.~Szmigelski (University of Saskatchewan,
Saskatoon) for their interest in the work and instructive discussions during the XXXV Workshop on Geometric Methods in Physics, held in Bia{\l}owie\.{z}a, Poland. The author is grateful to Professor B.~Kruglikov (University of Troms{\o}, Norway) for interest in the work and mentioning some misprints and important references, which were very helpful in preparing the manuscript. He is also indebted to the referees for their remarks and instrumental suggestions. Last but not least, thanks go to the Department of Mathematical Sciences of the NJIT (Newark NJ, USA) for the invitation to visit the NJIT during the Summer Semester of 2017, where an essential part of this paper was completed. Local support from the Institute of Mathematics at the Cracow University of Technology is also much appreciated.

\pdfbookmark[1]{References}{ref}
\LastPageEnding


\begin{thebibliography}{99}
\footnotesize\itemsep=0pt

\bibitem{BaKrMoVo}
Baran H., Krasil'shchik I.S., Morozov O.I., Voj\v{c}\'ak P., Integrability
 properties of some equations obtained by symmetry reductions,
 \href{https://doi.org/10.1080/14029251.2015.1023582}{\textit{J.~Nonlinear Math. Phys.}} \textbf{22} (2015), 210--232,
 \href{https://arxiv.org/abs/1412.6461}{arXiv:1412.6461}.

\bibitem{BlPrSaWaKu}
Blackmore D.L., Prykarpatska N.K., Samoylenko V.H., Wachnicki E., Pytel-Kudela
 M., The {C}artan--{M}onge geometric approach to the characteristic method for
 nonlinear partial differential equations of the first and higher orders,
 \href{https://doi.org/10.1007/s11072-007-0003-3}{\textit{Nonlinear Oscil.}} \textbf{10} (2007), 22--–31.

\bibitem{Bogd}
Bogdanov L.V., Interpolating differential reductions of multidimensional
 integrable hierarchies, \href{https://doi.org/10.1007/s11232-011-0055-5}{\textit{Theoret. and Math. Phys.}} \textbf{167}
 (2011), 705--713, \href{https://arxiv.org/abs/1011.0631}{arXiv:1011.0631}.

\bibitem{BoDrMa}
Bogdanov L.V., Dryuma V.S., Manakov S.V., Dunajski generalization of the second
 heavenly equation: dressing method and the hierarchy, \href{https://doi.org/10.1088/1751-8113/40/48/005}{\textit{J.~Phys.~A:
 Math. Theor.}} \textbf{40} (2007), 14383--14393, \href{https://arxiv.org/abs/0707.1675}{arXiv:0707.1675}.

\bibitem{BoKo}
Bogdanov L.V., Konopelchenko B.G., On the heavenly equation hierarchy and its
 reductions, \href{https://doi.org/10.1088/0305-4470/39/38/006}{\textit{J.~Phys.~A: Math. Gen.}} \textbf{39} (2006), 11793--11802,
 \href{https://arxiv.org/abs/nlin.SI/0512074}{nlin.SI/0512074}.

\bibitem{BoPa}
Bogdanov L.V., Pavlov M.V., Linearly degenerate hierarchies of quasiclassical
 {SDYM} type, \href{https://doi.org/10.1063/1.5004258}{\textit{J.~Math. Phys.}} \textbf{58} (2017), 093505, 13~pages,
 \href{https://arxiv.org/abs/1603.00238}{arXiv:1603.00238}.

\bibitem{BrGrGr}
Bruce A.J., Grabowska K., Grabowski J., Remarks on contact and {J}acobi
 geometry, \href{https://doi.org/10.3842/SIGMA.2017.059}{\textit{SIGMA}} \textbf{13} (2017), 059, 22~pages,
 \href{https://arxiv.org/abs/1507.05405}{arXiv:1507.05405}.

\bibitem{BuFeTs}
Burovskiy P.A., Ferapontov E.V., Tsarev S.P., Second-order quasilinear {PDE}s
 and conformal structures in projective space, \href{https://doi.org/10.1142/S0129167X10006215}{\textit{Internat.~J. Math.}}
 \textbf{21} (2010), 799--841, \href{https://arxiv.org/abs/0802.2626}{arXiv:0802.2626}.

\bibitem{Cart}
Cartan H., Differential forms, Houghton Mifflin Co., Boston, Mass, 1970.

\bibitem{DuNoFo}
Dubrovin B.A., Fomenko A.T., Novikov S.P., Modern geometry~--- methods and
 applications. {P}art~{I}. The geometry of surfaces, transformation groups,
 and fields, \href{https://doi.org/10.1007/978-1-4612-4398-4}{\textit{Graduate Texts in Mathematics}}, Vol.~93, 2nd ed.,
 Springer-Verlag, New York, 1992.

\bibitem{DuFeKr}
Dunajski M., Ferapontov E.V., Kruglikov B., On the {E}instein--{W}eyl and
 conformal self-duality equations, \href{https://doi.org/10.1063/1.4927251}{\textit{J.~Math. Phys.}} \textbf{56} (2015),
 083501, 10~pages, \href{https://arxiv.org/abs/1406.0018}{arXiv:1406.0018}.

\bibitem{DuKr}
Dunajski M., Kry\'nski W., Einstein--{W}eyl geometry, dispersionless {H}irota
 equation and {V}eronese webs, \href{https://doi.org/10.1017/S0305004114000164}{\textit{Math. Proc. Cambridge Philos. Soc.}}
 \textbf{157} (2014), 139--150, \href{https://arxiv.org/abs/1301.0621}{arXiv:1301.0621}.

\bibitem{FeKr}
Ferapontov E.V., Kruglikov B.S., Dispersionless integrable systems in 3{D} and
 {E}instein--{W}eyl geometry, \href{http://projecteuclid.org/euclid.jdg/1405447805}{\textit{J.~Differential Geom.}} \textbf{97}
 (2014), 215--254, \href{https://arxiv.org/abs/1208.2728}{arXiv:1208.2728}.

\bibitem{FeKrNo}
Ferapontov E.V., Kruglikov B.S., Novikov V., Integrability of dispersionless
 {H}irota type equations in {4D} and the symplectic {M}onge--{A}mp\`ere
 property, \href{https://arxiv.org/abs/1707.08070}{arXiv:1707.08070}.

\bibitem{GiTs-1}
Gibbons J., Tsarev S.P., Reductions of the {B}enney equations, \href{https://doi.org/10.1016/0375-9601(95)00954-X}{\textit{Phys.
 Lett.~A}} \textbf{211} (1996), 19--24.

\bibitem{GiTs-2}
Gibbons J., Tsarev S.P., Conformal maps and reductions of the {B}enney
 equations, \href{https://doi.org/10.1016/S0375-9601(99)00389-8}{\textit{Phys. Lett.~A}} \textbf{258} (1999), 263--271.

\bibitem{Godb}
Godbillon C., G\'eom\'etrie diff\'erentielle et m\'ecanique analytique,
 Hermann, Paris, 1969.

\bibitem{Grom}
Gromov M., Partial differential relations, \href{https://doi.org/10.1007/978-3-662-02267-2}{\textit{Ergebnisse der Mathematik
 und ihrer Grenzgebiete~(3)}}, Vol.~9, Springer-Verlag, Berlin, 1986.

\bibitem{HePrBlPr}
Hentosh O.E., Prykarpatsky Y.A., Blackmore D., Prykarpatski A.K., Lie-algebraic
 structure of {L}ax--{S}ato integrable heavenly equations and the
 {L}agrange--d'{A}lembert principle, \href{https://doi.org/10.1016/j.geomphys.2017.06.003}{\textit{J.~Geom. Phys.}} \textbf{120}
 (2017), 208--227, \href{https://arxiv.org/abs/1612.07760}{arXiv:1612.07760}.

\bibitem{Kiri}
Kirillov A.A., Local {L}ie algebras, \href{https://doi.org/10.1070/RM1976v031n04ABEH001556}{\textit{Russian Math. Surveys}} \textbf{31}
 (1976), no.~4, 55--75.

\bibitem{Kras}
Krasil'shchik I.S., A natural geometric construction underlying a class of
 {L}ax pairs, \href{https://doi.org/10.1134/S1995080216010054}{\textit{Lobachevskii~J. Math.}} \textbf{37} (2016), 60--65,
 \href{https://arxiv.org/abs/1401.0612}{arXiv:1401.0612}.

\bibitem{KrSeMo}
Krasil'shchik I.S., Sergyeyev A., Morozov O.I., Infinitely many nonlocal
 conservation laws for the {$ABC$} equation with {$A+B+C\neq0$}, \href{https://doi.org/10.1007/s00526-016-1061-0}{\textit{Calc.
 Var. Partial Differential Equations}} \textbf{55} (2016), 123, 12~pages,
 \href{https://arxiv.org/abs/1511.09430}{arXiv:1511.09430}.

\bibitem{MaSa}
Manakov S.V., Santini P.M., Cauchy problem on the plane for the dispersionless
 Kadomtsev--Petviashvili equation, \href{https://doi.org/10.1134/S0021364006100080}{\textit{JETP Lett.}} \textbf{83} (2006),
 462--466, \href{https://arxiv.org/abs/nlin.SI/0604023}{nlin.SI/0604023}.

\bibitem{Mokh}
Mokhov O.I., Symplectic and Poisson geometry on loop spaces of smooth manifolds
 and integrable equations, \textit{Contemporary Mathematics}, Institute for Computer
 Studies Publ., Izhevsk, 2004.

\bibitem{MoSe}
Morozov O.I., Sergyeyev A., The four-dimensional {M}art\'{\i}nez
 {A}lonso--{S}habat equation: reductions and nonlocal symmetries,
 \href{https://doi.org/10.1016/j.geomphys.2014.05.025}{\textit{J.~Geom. Phys.}} \textbf{85} (2014), 40--45, \href{https://arxiv.org/abs/1401.7942}{arXiv:1401.7942}.

\bibitem{OdSo}
Odesskii A.V., Sokolov V.V., Non-homogeneous systems of hydrodynamic type
 possessing {L}ax representations, \href{https://doi.org/10.1007/s00220-013-1787-x}{\textit{Comm. Math. Phys.}} \textbf{324}
 (2013), 47--62, \href{https://arxiv.org/abs/1206.5230}{arXiv:1206.5230}.

\bibitem{Serg-1}
Sergyeyev A., New integrable ({$3+1$})-dimensional systems and contact
 geometry, \href{https://doi.org/10.1007/s11005-017-1013-4}{\textit{Lett. Math. Phys.}} \textbf{108} (2018), 359--376,
 \href{https://arxiv.org/abs/1401.2122}{arXiv:1401.2122}.

\bibitem{SzBl}
Szablikowski B.M., B{\l}aszak M., Meromorphic {L}ax representations of
 {$(1+1)$}-dimensional multi-{H}amiltonian dispersionless systems,
 \href{https://doi.org/10.1063/1.2344853}{\textit{J.~Math. Phys.}} \textbf{47} (2006), 092701, 23~pages,
 \href{https://arxiv.org/abs/nlin.SI/0510068}{nlin.SI/0510068}.

\bibitem{TaTa-1}
Takasaki K., Takebe T., {${\rm SDiff}(2)$} {T}oda equation~-- hierarchy, tau
 function, and symmetries, \href{https://doi.org/10.1007/BF01885498}{\textit{Lett. Math. Phys.}} \textbf{23} (1991),
 205--214, \href{https://arxiv.org/abs/hep-th/9112042}{hep-th/9112042}.

\bibitem{TaTa-2}
Takasaki K., Takebe T., Integrable hierarchies and dispersionless limit,
 \href{https://doi.org/10.1142/S0129055X9500030X}{\textit{Rev. Math. Phys.}} \textbf{7} (1995), 743--808,
 \href{https://arxiv.org/abs/hep-th/9405096}{hep-th/9405096}.

\bibitem{Zakh}
Zakharevich I., Nonlinear wave equation, nonlinear {R}iemann problem, and the
 twistor transform of {V}eronese webs, \href{https://arxiv.org/abs/math-ph/0006001}{math-ph/0006001}.

\bibitem{ZaSh}
Zakharov V.E., Shabat A.B., Integration of nonlinear equations of mathematical
 physics by the method of inverse scattering.~{II}, \href{https://doi.org/10.1007/BF01077483}{\textit{Funct. Anal.
 Appl.}} \textbf{13} (1979), 166--174.

\end{thebibliography}
\end{document}